\documentclass[sigconf]{aamas}

\usepackage{algpseudocode}

\usepackage{balance} % for balancing columns on the final page

\usepackage{algorithm}
\usepackage{thmtools}
\usepackage{thm-restate}
\newtheorem{theorem}{Theorem}
\newtheorem{lemma}{Lemma}
\newtheorem{remark}{Remark}
\newtheorem{definition}{Definition}
\newtheorem{corollary}{Corollary}
\usepackage{cleveref}
\renewcommand{\P}{\mathcal{P}}
\newcommand{\s}{\textbf{s}}

\setcopyright{none}
\renewcommand{\footnotetextcopyrightpermission}[1]{}

\acmBooktitle{}
\acmConference{}{}{}{}

\title[Greedy Routing Reachability Games]{Greedy Routing Reachability Games}

\author{Pascal Lenzner}
\affiliation{
  \institution{University of Augsburg}
  \city{Augsburg}
  \country{Germany}}
\email{pascal.lenzner@uni-a.de}

\author{Paraskevi Machaira}
\affiliation{
  \institution{Hasso Plattner Institute}
  \city{Potsdam}
  \country{Germany}}
\email{paraskevi.machaira@hpi.de}

\begin{abstract}
Today's networks consist of many autonomous entities that follow their own objectives, i.e., smart devices or parts of large AI systems, that are interconnected. Given the size and complexity of most communication networks, each entity typically only has a local view and thus must rely on a local routing protocol for sending and forwarding packets. A common solution for this is greedy routing, where packets are locally forwarded to a neighbor in the network that is closer to the packet's destination. 

In this paper we investigate a game-theoretic model with autonomous agents that aim at forming a network where greedy routing is enabled. The agents are positioned in a metric space and each agent tries to establish as few links as possible, while maintaining that it can reach every other agent via greedy routing. Thus, this model captures how greedy routing networks are formed without any assumption on the distribution of the agents or the specific employed greedy routing protocol. Hence, it distills the essence that makes greedy routing work. 

We study two variants of the model: with directed edges or with undirected edges. For the former, we show that equilibria exist, have optimal total cost, and that in Euclidean metrics they can be found efficiently. However, even for this simple setting computing optimal strategies is NP-hard. For the much more challenging setting with undirected edges, we show for the realistic setting with agents in 2D Euclidean space that the price of anarchy is between 1.75 and 1.8 and for higher dimensions it is less than 2. Also, we show that best response dynamics may cycle, but that in Euclidean space almost optimal approximate equilibria can be computed in polynomial time. Moreover, for 2D Euclidean space, these approximate equilibria outperform the well-known Delaunay triangulation. 
\end{abstract}


\newcommand{\BibTeX}{\rm B\kern-.05em{\sc i\kern-.025em b}\kern-.08em\TeX}

\begin{document}

\pagestyle{fancy}
\fancyhead{}

\maketitle

\section{Introduction}
Greedy routing is widely used in many applications, because it enables communication between the nodes of a network without requiring complete knowledge of the network's structure. Given that our communication networks today are large, complex, and potentially also dynamically changing, using traditional shortest path based routing, facilitated via fixed routing tables that must be distributed and maintained at all important network nodes, will have to be replaced by local and more adaptive alternatives. One of the most prominent options for this is greedy routing, also known as geographic routing or stateless routing \cite{finn1987routing, karp2000gpsr}. While common for small ad-hoc peer-to-peer networks of smart devices or sensors, greedy routing was recently proposed for the whole Internet~\cite{boguna2010sustaining}. 

The key for greedy routing is that every network node has a position in some underlying space, typically a metric space like the Euclidean plane or the hyperbolic disc. These positions allow for deciding locally which network neighbor to use as the next hop in packet routing. If at every step of the greedy routing path always a next hop exists that is strictly closer to the target position than the current node then greedy routing succeeds. If greedy routing succeeds between all pairs of nodes, the network is called \emph{navigable}. Thus, the challenge is to create a navigable network. 

Significant research has been conducted on greedy routing but mostly with the goal of finding a suitable greedy embedding, i.e., a mapping of nodes to virtual coordinates in some space to facilitate routing~\cite{WestphalP09,CvetkovskiC09,Blasius0KK20}. For this, a fixed network is given and the coordinates from the embedding guide packets on their way to their destination. A different approach and equally important, is the setting where there is no predefined network and the nodes have to construct a network among themselves such that greedy routing is enabled. One prominent application are first responder networks in areas where a natural disaster like an earthquake or tsunami has destroyed the communication infrastructure. The idea is that communication sensors are dropped in the area that then establish a communication network that can be used by rescue teams and disaster relief workers~\cite{firstresponder,firstresponderKumarRS04}. For such applications it is typically assumed that the process of setting up the network is coordinated in the sense that all devices run the same local algorithm.  Another application is database similarity search where constructing navigable networks is mentioned as an open problem~\cite{diwan2024navigable, al2025distance}.

While embeddings work with a fixed network and try to find suitable node positions, shifting the focus away from greedy embeddings highlights two important aspects of the problem. First, and ever more important in the current age of smart devices, it allows the study of a decentralized setting without a central authority dictating the network structure. Instead, following a game-theoretic approach~\cite{Papadimitriou01} and in particular the well-known paradigm of network creation games~\cite{bala2000noncooperative,fabrikantNetwork2003}, the nodes themselves make strategic decisions on which edges to build in order to enable greedy routing, calling for a game-theoretic analysis. Second, removing the constraints imposed by a given network reveals the fundamental mechanisms and structural properties that greedy routing networks necessarily require. This understanding is essential for designing optimal networks that enable greedy routing and that have certain structural properties, such as minimizing the total cost.

In this paper we set out to investigate the structural properties of navigable networks that arise from the interaction of selfish agents. For this, we consider a very basic network formation game, where agents that correspond to nodes with a fixed position in a metric space strategically set up links to other nodes to enable greedy routing paths to all network nodes. This objective aims at greedy routing reachability, i.e., the most basic property needed for all agents. With this, we lay the foundation for more complex models that besides greedy navigability might add more objectives, like robustness or stretch guarantees. 
We find that almost stable states exist and can be efficiently computed. Moreover, these equilibria have favorable properties since they are guaranteed to be close to centrally optimized networks in terms of cost, i.e., in terms of the number of built edges. Smart agents can make independent decisions. We show that such agents do not need central coordination to create almost optimal navigable networks.

\subsection{Model and Preliminaries}

We consider $n$ agents corresponding to points $\P = \{p_1,\dots,p_n\}$ in a metric space $\mathcal{M} = (\mathcal{P},d_\mathcal{M})$, where $d_\mathcal{M}(u,v)$ denotes the distance between any two points $u, v\in \mathcal{P}$, with $d_\mathcal{M}(u,u)=0$ for all $u\in \mathcal{P}$, and where the triangle inequality holds. If clear from the context, we omit $\mathcal{M}$. The agents, knowing the positions of all points, aim to create a network that enables greedy routing. For this, each agent decides on a set of incident edges it wants to establish. 

The \emph{strategy of agent~$u \in \P$} is $S_u$, a set of edges with endpoints in $\P\setminus{u}$, i.e., agent~$u$ can buy incident edges to any subset of other agents. For each edge in $S_u$, we say that agent $u$ \emph{owns} the edge.
We consider two variants. In the directed case, each edge in $S_u$ is of the form $(u,v)$, representing that agent $u$ buys a directed edge to $v$. In the undirected case, each edge in $S_u$ is of the form $\{u,v\}$, denoting that agent $u$ buys an undirected edge to $v$. Let $\s = (S_1,\dots,S_n)$ be the \emph{strategy-profile}, i.e., the vector of strategies of all agents. Also, for agent $u\in \P$ let $\s = (S_u,\s_{-u})$, where $\s_{-u}$ is the vector of strategies of all agents except $u$. We sometimes use the notation $S_u(\s)$ to explicitly refer to the strategy profile which contains strategy $S_u$.
Any strategy-profile~$\s$ defines a network $G(\s) = \left(\mathcal{P}, E(\s)\right)$,  where $E(\s) = \bigcup_{u\in \P}S_u$ in case of directed edges or $ E(\s) = \bigcup_{u\in \P}\{\{u,v\} \mid \{u,v\}\in S_u \vee \{u,v\}\in S_v\}$ in case of undirected edges. Vertices of $G(\s)$ are called agents, nodes, or points.

Given a network $G = (\mathcal{P},E)$, a \emph{greedy routing path from \(u\) to \(v\) in $G$} is a sequence \((x_1, x_2,\dots , x_j)\), with $x_i\in \mathcal{P}$, for $1\leq i \leq j$, where \(x_1 = u\), \(x_j = v\), and $(x_i,x_{i+1}) \in E$ (or $\{x_i,x_{i+1}\} \in E$, for undirected edges), for $1\leq i \leq j-1$, such that \(d(x_i, v) > d(x_{i+1}, v)\) holds for all $1\leq i\leq j-1$. 
Thus, such a path is a directed (undirected) path from $u$ to $v$ in $G$, where along the path the nodes get strictly closer to the endpoint of the path in terms of their distance. We say that \emph{greedy routing is enabled} 
for an agent~$u$ in network~$G$, or that $u$ is \emph{greedy connected}, if in $G$ there exists a greedy routing path from $u$ to every other node. If this holds for every node in $G$, then we say that greedy routing is enabled in $G$. See \Cref{fig:greedy-example}.
\begin{figure}[t]
    \centering
    \includegraphics[width=0.4\linewidth]{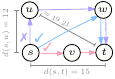}
    \caption{Illustration of greedy routing paths in $\mathbb{R}^2$. The pink and blue $s$-$t$-paths are greedy routing paths, while the purple $s$-$t$-path is not, since $ 15 = d(s,t) < d(u,t) = \sqrt{12^2+15^2}$. Greedy routing is enabled only for node $s$.}
    \label{fig:greedy-example}
\end{figure}

Agents select strategies to minimize their \emph{cost} within the formed network. The cost of agent~\(u\) in network $G(\mathbf{s})$ is defined as $$c_u(\mathbf{s}) = |S_u| + \omega,$$
where $\omega$ serves as a penalty term, i.e., it is equal to $0$ if greedy routing is enabled for agent~$u$ and $\infty$ otherwise.
The \emph{social cost} of network $G(\mathbf{s})$ is $c(\mathbf{s}) = \sum_{u\in \mathcal{P}}c_u(\mathbf{s})$, i.e., the total cost over all agents. For a set of points~$\mathcal{P}$, the network $G(\mathbf{s}^*) = (\mathcal{P},E(\mathbf{s}^*))$ minimizing the social cost is called the \emph{social optimum network (SO)} for~$\mathcal{P}$.

A strategy $S^*_u$ is called a \emph{best response} of agent~\(u\) in strategy-profile $\mathbf{s} = (S_u,\mathbf{s}_{-u})$, if $c_u(S^*_u,\mathbf{s}_{-u}) \leq c_u(S'_u,\mathbf{s}_{-u})$ for any other strategy $S'_u$ where 
in the directed case $S'_u \subseteq \{(u,v)\mid v\in \mathcal P,\ v\neq u\}
$, while in the undirected case $S'_u \subseteq \{ \{u,v\}\mid v\in \mathcal P,\ v\neq u\}$. That is, a best response strategy minimizes agent $u$'s cost given that the other agents' strategies are fixed. 
A strategy-profile~\(\mathbf{s}\) is a \emph{pure Nash Equilibrium} (NE), or \emph{stable} for short, if in profile~$\mathbf{s}$ every agent already plays a best response. We denote by NE the set of all strategy profiles that are pure Nash equilibria. For the setting with directed edges, we have a bijection between strategy-profiles~$\mathbf{s}$ and the corresponding networks $G(\mathbf{s})$. In the case of undirected edges, the same holds if also the edge ownership information is considered. In both cases, we will say that a network $G(\mathbf{s})$ is in NE, if the strategy profile~$\mathbf{s}\in$ NE\footnote{For undirected edges, in NE no edge can be bought by both endpoints.}. Note that in every NE network greedy routing must be enabled, since every agent could buy direct edges to every other node to avoid the infinite penalty $\omega$.

We also consider weaker versions of stability. A network \(G(\mathbf{s})\) is in \emph{\(\beta\text{-approximate}\) NE} ($\beta$-NE or \emph{$\beta$-stable}) if no agent \(u\) can change its strategy such that its cost decreases below \(\frac{1}{\beta} c_u(\mathbf{s})\), i.e., no agent can reduce its cost to a $\beta$-fraction by any strategy change. Similarly, $G(\s)$ is in \emph{$\gamma\text{-additive}$ NE} ($+\gamma$-NE or \emph{$+\gamma$-stable}) if no agent $u$ can decrease its cost below $c_u(\s)-\gamma$. If greedy routing is enabled, then in a $\beta$-stable network, agents buy at most a factor of $\beta$ too many edges, while in a $+\gamma$-stable network they buy at most $\gamma$ extra edges. Thus, $1$-stable and $+0$-stable networks are in NE. 

The Price of Anarchy (PoA)~\cite{KP09} is defined as the worst-case ratio of the social cost of any NE and the corresponding social optimum. 

A \emph{best response path} is a sequence $\mathbf{s}_0,\mathbf{s}_1,\dots,\mathbf{s}_k$ of strategy-profiles, such that $\mathbf{s}_i$ results from some agent switching to a best response in $\mathbf{s}_{i-1}$, for $1\leq i\leq k$. A \emph{best response cycle} (BRC) is a cyclic best response path, i.e., $\mathbf{s}_0 = \mathbf{s}_k$, and its existence implies that a game cannot have an ordinal potential function~\cite{monderer1996potential}. Hence, it is not a potential games for which the existence of a  NE is guaranteed.

For each point $u\in \P$ let the set of \emph{nearest neighbors} of $u$ be $N(u) = \arg\min_{v\in\P\setminus\{u\}}{d_{\mathcal{M}}(u,v)}$. Thus, the set~$N(u)$ consists of points in $\P$ with minimum distance to $u$. The \emph{Nearest Neighbor Graph (NNG)} on $\P$, denoted by $G^{NNG}(\P, E)$ is defined as follows. In the directed case, the edge set is $E = \{(u, v) \mid v \in N(u)\}$, i.e., each point has a directed edge to its nearest neighbors, while in the undirected case,  $E = \{\{u, v\} \mid v \in N(u) \vee u\in N(v)\}$. See \Cref{fig:nng}~(left) for an example. It is known that the NNG has to be a subgraph of every graph that supports greedy routing \cite{PapadimitriouR05}. 

A \emph{Delaunay triangulation ($DT$)} \cite{DBLP:books/lib/BergCKO08} on a set of points~$\P$ in the $D$-dimensional Euclidean space is a triangulation such that no point in $\P$ is inside the circum-hypersphere of any $D$-simplex in $DT(\P)$ (see \Cref{fig:nng}~right) for an example. An important structural property is that the NNG on $\P$ is a subgraph of $DT(\P)$. Moreover, it has been shown in \cite{Bose_Morin_2004} that Delaunay triangulations in two-dimensional Euclidean space support greedy routing. For this case, the Delaunay triangulation is a planar graph, i.e., it has at most $3|\mathcal{P}|-6$ edges.

\begin{figure}
    \centering
\includegraphics[width=0.7\linewidth]{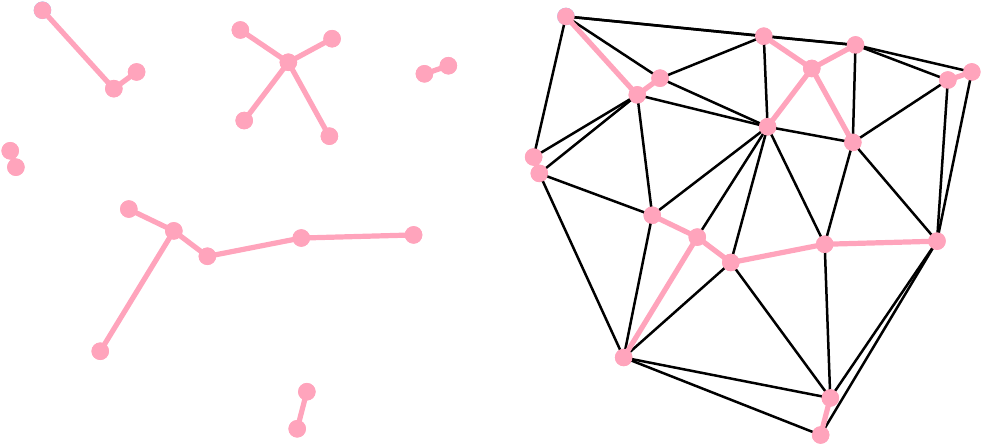}
    \caption{(right) The nearest neighbor graph of 20 points $\mathcal{P}$ in $\mathbb{R}^2$, (left) The Delaunay triangulation of $\mathcal{P}$.}
    \label{fig:nng}
\end{figure}

The \emph{kissing number} in $\mathbb{R}^D$, denoted by $K(D)$ is the maximum number of disjoint unit hyperspheres that can simultaneously touch a given unit hypersphere in $\mathbb{R}^D$. Equivalently~\cite{conway2013sphere}, $K(D)$ is the maximum number of points that can be placed on the surface of a unit sphere in $\mathbb{R}^D$ so that the angular distance between any two points is at least $60^\circ$. We have $K(D) \leq 2^{0.401D}$ but exact values of $K(D)$ are known only for $D\in\{1, 2, 3, 4, 8, 24\}$. See \Cref{fig:kissing_number}. 
\begin{figure}[h]
    \centering
    \includegraphics[width=0.18\linewidth]{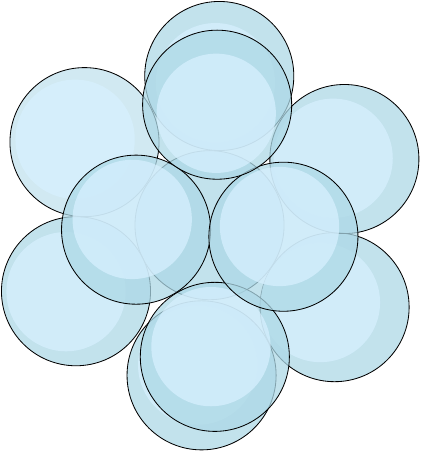}
    \caption{The kissing number in $\mathbb{R}^3$ is 12.}
    \label{fig:kissing_number}
\end{figure}

For an agent $u\in\P$ embedded in a metric space, we define the \emph{greedy routing set} as a set of edges, such that, after connecting with these edges, agent $u$ has at least one neighbor $v$ satisfying
$d(v,w)<d(u,w)$, for all $w\in \P\setminus\{u\}$. A \emph{minimum greedy routing set}, denoted $\Phi(u)$, is the smallest such set for an agent $u$, and its size $\phi(u)=|\Phi(u)|$ is called the \emph{greedy routing degree} of agent $u$. In other words, $\phi(u)$ measures the minimum connectivity effort required for agent~$u$ to locally enable greedy routing. By $\phi^+(u)$, we denote the number of these edges that belong to the NNG. See \Cref{fig:phi}.

\begin{figure}[h]
    \centering
    \includegraphics[width=\linewidth]{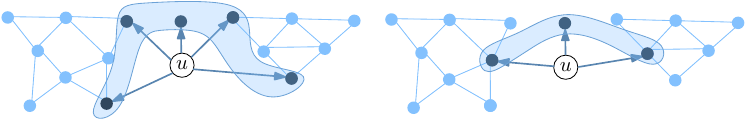}
    \caption{(left) A greedy routing set of agent $u$, (right) the minimum greedy routing set with $\phi(u)=3$ and $\phi^+(u)=1$.}
    \label{fig:phi}
\end{figure}

\subsection{Related Work}
In greedy routing \cite{finn1987routing}, each node forwards a packet to the neighbor that is geographically closest to the destination, using only local information. However, a packet may get stuck at a \emph{local minimum} in the network. Combining greedy forwarding with \emph{perimeter routing} to recover from such local minima is introduced independently by Karp and Kung \cite{karp2000gpsr} and Bose et al. \cite{bose1999routing}. Subsequent studies extended this approach to handle routing voids and to improve delivery guarantees \cite{fang2006locating, liu2008greedy, datta2002internal, kuhn2003worst}. Another method for overcoming voids is assigning virtual coordinates to the nodes, constructing a greedy embedding \cite{kleinberg2007geographic, eppstein2008succinct, Blasius0KK20}. While such embeddings can be computed efficiently, they introduce limitations and do not express real distances. Despite the extensive research on greedy routing, not much is known for the problem of constructing navigable networks.

The Delaunay triangulation (DT)~\cite{DBLP:books/lib/BergCKO08} plays a central role in navigable networks as it enables greedy routing in $\mathbb{R}^2$ \cite{Bose_Morin_2004} and $\mathbb{R}^3$ \cite{wang2013three}. Also, if the destination can be any point in the continuous space rather than a discrete set of nodes, the DT without degenerate edges is a subgraph of every graph that supports greedy routing~\cite{Ghaffari_Hariri_Shirmohammadi_2010}.

Navigable networks in higher dimensional spaces have attracted interest because of their application in distributed similarity search and high-dimensional data indexing \cite{diwan2024navigable, fu2017fast}. Structures like hierarchical navigable small-world networks~\cite{malkov2018efficient} use greedy routing principles to perform efficient approximate nearest-neighbor search. Most of the work on this is focused on empirical studies.

Since navigable networks share many structural properties with real-world complex networks \cite{boguna2009navigability}, game theory provides a useful framework for understanding how such networks form and evolve over time. Bala and Goyal~\cite{bala2000noncooperative} introduced a non-cooperative network formation game where agents correspond to network nodes and buy incident edges to reach every other agent. In contrast to our model, any path connecting two nodes establishes reachability. They show that NE exist and can be easily found via suitable game dynamics. More involved are the network creation games (NCGs) by Fabrikant et al.~\cite{fabrikantNetwork2003}, where the agents buy incident undirected edges to minimize the sum of hop distances to all other agents. They show that NE exist and have a low PoA, but computing a best response is NP-hard. NCG variants that involve geometry are proposed by Eidenbenz et al.~\cite{eidenbenzEquilibria2006}, where agents corresponding to points in the plane establish undirected edges to ensure reachability, by Abam et al.~\cite{abam2019geometric}, where the agents aims for a low maximum stretch (defined as the ratio of graph distance to metric space distance), and by Bilò et al.~\cite{biloGeometric2019} and Friedemann et al.~\cite{friedemannEfficiency2021}, where agents aim to minimize their sum of shortest path distances to the other agents.

The works most closely to ours, where greedy routing is studied within the network creation framework, are those of Gulyás et al.~\cite{gulyasNavigable2015} and Berger et al.~\cite{berger2025strategic}. In \cite{gulyasNavigable2015} agents aim to  reach other agents via greedy routing in hyperbolic space. They study the directed case where agents do not influence each other and prove that computing a minimum greedy routing set is NP-hard. In \cite{berger2025strategic},  each connection has a cost $\alpha>0$, and agents aim to create a directed network in which greedy routing works,  while minimizing connection costs and the sum of stretches to the other agents. Most results are limited to 1-2 metrics (where $d_{\mathcal{M}}(u,v)\in \{1,2\}$ for any $u,v\in\P$) and tree metrics, although they prove the existence of a 5-approximate NE in $\mathbb{R}^2$, and an $(\alpha+1)$-NE in general metric spaces. 

\subsection{Our Contribution}
Our work establishes the theoretical foundations that bridge game theory and computational geometry,  introducing tools for developing more complex models for network creation. While most of the literature focuses on greedy embeddings or on greedy routing in families of pre-constructed networks with directed edges, we study the process of creating networks, directed or undirected, that enable greedy routing. In contrast to existing work, our model captures how greedy routing networks are formed without assumptions on the distribution of the agents or the specific greedy routing protocol. Hence, it distills the essence that makes greedy routing work. 

In particular, for agents in 2D Euclidean metrics we give an algorithm that efficiently constructs a navigable network that is almost stable from a game-theoretic point of view, while having at most $80\%$ more edges than the sparsest possible navigable network. This improves on the approximation factor of $3$ achieved by the well-known Delaunay triangulation. Moreover, this can be seen as the basic step needed for designing efficient algorithms for creating networks that incorporate additional aspects like stretch or robustness guarantees. The created networks are useful for applications such as sensors that are randomly dropped over an area, IoT networks, and other situations where a communication network has to be formed among smart devices. Also, our approach adapts well to dynamic network topologies, is independent of the employed routing protocol, and it is more robust, since it only relies on local communication between immediate neighbors.

While focusing on 2D Euclidean space, we show that our approach also works in general metrics and higher dimensions, yielding navigable networks with at most twice as many edges as the optimum. This increase is due to agent selfishness, and we study this by giving almost tight bounds on the price of anarchy.

\section{Properties for Greedy Routing}
We establish some structural properties that will be used in our analysis of both directed and undirected variants. These properties are essential for enabling greedy routing. 
\paragraph*{\textbf{General Metrics:}} We start with properties that hold for any metric space. Thus, these properties are the foundation for greedy routing in any network. 
\begin{restatable}{lemma}{lemmaalphaone}
\label{independent}
    Consider two agents $u,v\in \mathcal{P}$ connected by an edge $(u,v)$ (or an undirected edge $\{u,v\}$). If greedy routing is enabled for agent~$v$, then the agents that can be reached by a greedy routing path from $u$ via $v$ are independent of agent~$v$'s strategy.   
\end{restatable}
\begin{proof}
Consider any destination point $w\in \mathcal{P}$ and two strategy profiles $\s, \s'$ in which greedy routing is enabled for agent~$v$ but where $S_v \neq S_v'$. Suppose towards a contradiction that there is a greedy routing path from agent~$u$ to $w$ via node~$v$ in profile~$\s$ while node~$w$ is not reachable by greedy routing via node~$v$ in profile~$\s'$.

Since greedy routing is enabled for agent~$v$ in $\s'$, there is a greedy routing path $P=(v_i, v_{i+1},\dots, v_k)$ of length $k$ from $v$ to $w$, where $v=v_0, v_k=w$ and $0 \leq i < k$. According to the definition of greedy routing, we have that $d(v_i, v_k) > d(v_{i+1}, v_k)$, for every $0 \leq i<k$. Since there is a greedy routing path from $u$ to $w$ via $v$ in profile~$\s$, we have that $d(v,w) < d(u,w)$. Thus, the path $(u,v)\cup P$ is a greedy routing path in $\s'$, which is a contradiction.
\end{proof}
Next, we have another rather obvious property.
\begin{restatable}{lemma}{lemmaalphatwo}
\label{lem:phi_edges}
    In a NE, an agent~$u$ builds at most $\phi(u)$ many edges.
\end{restatable}
\begin{proof}
By definition, the greedy routing degree~$\phi(u)$ is the smallest number of neighbors that agent~$u$ needs to be connected with in order to be greedy connected. In particular, there exists a set $W$ of $\phi(u)$ many nodes such that for every $t\in \P$ there is a $w\in W$ with $d(w,t)<d(u,t)$. 

Suppose towards a contradiction that in a NE network $G(\s)$ agent~$u$ with strategy $S_u$ builds $k>\phi(u)$ edges. Consider an alternative strategy $S_u'$ where agent~$u$ builds edges only to the nodes in $W$ and consider the corresponding strategy profile $\s'$, that differs from profile~$\s$ only in agent~$u$'s strategy. If agent~$u$ is greedy connected in $G(\s')$ then agent~$u$ could reduce its cost by choosing strategy $S_u'$ in network $G(\s)$, which is a contradiction to $G(\s)$ being a NE network. Therefore, in network $G(\s')$, where agent~$u$ buys edges only to nodes in $W$, there must be a node~$v$ that is not reachable from agent~$u$ via greedy routing. Thus, by definition of the set $W$, there exists a node $w\in W$ with $d(w, v)< d(u,v)$ that does not have a greedy routing path to node~$v$ in network $G(\s')$. Moreover, the greedy routing path from $u$ to $v$ in network $G(\s)$ cannot be used by agent~$w$, since $d(w, v) < d(u,v)$. Therefore, agent~$w$ is not greedy connected, thus, network $G(\s)$ is not a NE instance.
\end{proof}

\begin{remark} \label{re:consant_phi}
In a NE network, for any agent $u\in \P$ the value of $\phi^+(u)$ is independent of the strategy profile, since all edges of the NNG have to be present in any NE network. 
\end{remark}

\paragraph*{\textbf{Euclidean Metric:}} Now we move to properties that hold for Euclidean metrics. The following two lemmas were proven by \citet{Clarkson87}. For completeness, we provide the proofs in our context.

\begin{restatable}{lemma}{lemmaalphathree}
\label{cosine-rule}
In Euclidean metrics, consider nodes $s, u, v \in \mathcal{P}$, where $d(s, u) \leq d(s, v)$, and let $\theta$ denote the angle $\angle usv$. If $\theta < \pi/3$ then the sequence $P=(s, u, v)$ is a greedy routing path from node~$s$ to node~$v$.
\end{restatable}
\begin{proof}
In order to show that $P$ is a greedy routing path from node~$s$ to node~$v$, it suffices to show that $d(u,v)<d(s,v)$. Applying the Law of Cosines in triangle $vsu$, we have:
\begin{figure}
    \centering
    \includegraphics[width= 1\linewidth]{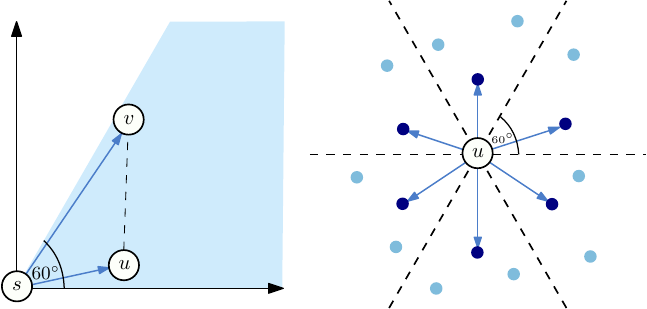}
    \caption{(left) The path from $s$ to $v$ via $u$ is a greedy routing path since $d(u,v)<d(s, v)$. (right) by building 6 edges, $u$ has a neighbor that is closer to every node than $u$ itself.}
    \label{fig:degrees}
\end{figure}

\begin{align*}
    d^2(u,v) &=  d^2(s,u)+d^2(s,v)-2\cdot d(s,u)\cdot d(s,v)\cdot cos(\theta)  \\
     &<  d^2(s,u)+d^2(s,v)-2\cdot d(s,u)\cdot d(s,v)\cdot 1/2 \\
      &= d^2(s,u)+d^2(s,v)-d(s,u)\cdot d(s,v)\\
     &< d^2(s,u)+d^2(s,v)-d^2(s,u)\\
      &= d^2(s,v)\\
    d(u,v)  &< d(s,v)
\end{align*}
Since $d(u,v) < d(s, v)$, the path from $s$ to $v$ via $u$ is a greedy routing path (see \Cref{fig:degrees}~(left) for an example).
\end{proof}

\begin{restatable}{lemma}{lemmaalphafour}
For 2D Euclidean metrics, the greedy routing degree of any agent can be computed in polynomial time and is at most~6.\label{sixedges}
\end{restatable}
\begin{proof}
We will prove this by construction (see \Cref{fig:degrees}): Partition the space around node~$u$ into six cones, each having an angle of $\pi/3$. In each cone select the node $v$ that is closest to $u$ and build the edge $(u,v)$. By \Cref{cosine-rule}, we have that agent~$u$ has a greedy routing path to every node in the cone via node~$v$. Therefore, the greedy routing degree~$\phi(u)$ is at most $6$.

To compute the greedy routing degree, we can simply try every subset of at most six edges, which takes polynomial time.
\end{proof}

\begin{remark}
    For Euclidean metrics in dimension $D>0$, the greedy routing degree of any agent can be computed in polynomial time and it is at most the kissing number $K(D)$. 
\end{remark}
The above follows by the definition of the kissing number by \citet{conway2013sphere} and by \Cref{cosine-rule}.

\section{Greedy Routing: Directed Edges}
In this section we explore the directed variant of our model beginning with some results that hold in every metric space.
\subsection{Nash Equilibria}
\paragraph*{\textbf{General Metrics.}} We start by proving that Nash equilibria always exist for arbitrary metric spaces.
\begin{restatable}{theorem}{theorembetaone}
\label{socialoptimum}
Every NE is a SO and every SO is a NE.
\end{restatable}
\begin{proof}
    By \Cref{independent}, we have that for every agent $u\in \P$, as long as agent~$u$ is greedy connected, agent~$u$'s strategy does not influence the other agents. Using this property we will show that every social optimum is a NE and every NE is a social optimum.

    First, assume towards a contradiction that there exists a NE that is not a social optimum. Then, there must be an agent~$u$ whose social cost is lower in the social optimum than in the NE. By deviating to the strategy used in the social optimum, agent~$u$ could reduce its cost, contradicting that the original strategy profile was a NE.

    Next, suppose that there is a social optimum that it is not a NE. Then, there is an agent~$u$ that could reduce its cost by changing its strategy. Since, such a strategy change does not influence the other agents, the total social cost decreases contradicting the assumption that the original strategy profile was a social optimum.
\end{proof}
\noindent This directly implies the following positive result.
\begin{restatable}{corollary}{betacoltwo}
    The price of anarchy in the directed version is 1.
\end{restatable}
Next, we show that best response cycles cannot exist. This yields that stable states can be found via natural game dynamics.
\begin{restatable}{lemma}{betacolthree}
Best response cycles do not exist.
\end{restatable}
\begin{proof}
For a best response cycle to exist, there must be at least one agent $u$, that in some round $r$ plays a strategy that it had played in a previous round. 

Let $S_u$ denote the strategy of $u$ at the beginning of round $r$, and let $S_u'$ denote the previous strategy that $u$ is willing to play again. Without loss of generality, suppose that $S_u$ was not the initial strategy that was assigned to agent $u$. Suppose that $|S_u|\leq|S_u'|$. In this case, agent~$u$ changes its strategy to a strategy with at least as many edges as its current strategy. An agent $u$ would only choose a strategy with more or the same number of edges if although $u$ was greedy connected when using strategy $S_u$ in a previous round, it is no longer greedy connected under $S_u$  in round $r$. 

Suppose towards a contradiction, that there exists an agent $w$ that was reachable from $u $ by a greedy routing path in $S_u$ in a previous round, but it is no longer reachable in round $r$. Then, there must exist an agent $x$ on that greedy routing path from $u$ to $w$ whose strategy changed in some round before $r$. Since, agent $x$ has no incentive to change its strategy to a strategy that does not enable greedy routing, its new strategy must continue to support greedy routing. By \Cref{independent}, we have that even after the strategy change of agent $x$, a greedy routing from $u$ to $w$ via $x$ still exists, contradiction. The argument is still valid if more than one agents along the greedy routing path from $u$ to $v$ change their strategy.

Therefore, all greedy routing paths from $u$ using edges in $S_u$ remain valid in round $r$, so agent $u$ has no incentive to change its strategy to a strategy with more edges or the same number of edges.

Now, consider the case where $|S_u| > |S_u'|$. Since $S_u'$ is a previously played strategy, agent $u$ must have changed its strategy from $S_u'$ to a strategy with more edges in an earlier round. By the argument above, such a change cannot occur, i.e., an agent has no incentive to change its strategy to a strategy that uses more edges. Thus, this case also leads to a contradiction. 

Therefore, there is no agent $u$ that changes its strategy to a strategy that it had before, i.e., best response cycles cannot exist. 

If $S_u'$ was the initial strategy assigned to $u$, the only possible exception occurs when $S_u'$ would have been a best response, but some agents associated with the edges in $S_u'$ where not greedy connected initially. In that case, $u$ changes its strategy in order to become greedy connected and when the agents associated with the edges in $S_u'$ become greedy routing connected, then $u$ changes its strategy back to $S_u'$. This can happen at most once for every agent $u$. After this, the previous arguments apply, showing that best response cycles cannot exist.  
\end{proof}
\subsection{Computational Complexity}
\paragraph*{\textbf{General Metrics}} Here, we investigate the computational complexity of computing a best response in general metric spaces, and computing a NE in Euclidean metrics. 

Although computing the greedy routing degree can be done in polynomial time, computing a best response is a hard problem, even in one-dimensional metric spaces.
\begin{restatable}{theorem}{betatheortwo}
    For Euclidean metrics (hence also for general metrics) computing a best response of an agent is NP-hard.\label{bestresponsehardness}
\end{restatable}

\begin{proof}
\begin{figure}[h]
    \centering
    \includegraphics[width=\linewidth]{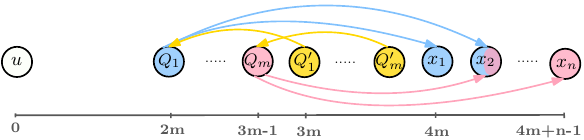}
    \caption{Illustration of the reduction from \textsc{set cover}.}
    \label{fig:bestresponse}
\end{figure}
We will reduce from the NP-hard problem $\textsc{set cover}$. Let $P=(\{x_1,\dots, x_n\}, \{Q_1,\dots, Q_m\})$ be an instance of $\textsc{set cover}$ where $\{x_1,\dots, x_n\}$ is the set of elements that need to be covered and $ \{Q_1,\dots, Q_m\}$ is a collection of subsets of the elements that can be selected for covering elements, i.e., we have $\{x_1,\dots,x_n\} \subseteq Q_1 \cup \dots \cup Q_m$. Given a \textsc{set cover}-instance, we construct a directed graph $G(V,E)$ as follows (see \Cref{fig:bestresponse}). We create a node for every set $Q_i$ and for every element $x_j$, for $1 \leq i\leq m$ and  $1 \leq j\leq n$. Moreover, we add an auxiliary node $Q'_i$ for every node $Q_i$, and additionally a single node~$u$. The nodes are placed on the line, so that node~$u$ is at position $0$, each $Q_i$ is at position $2m+i-1$, each $Q'_i$ is at position $3m+i-1$ and each $x_j$
is at position $4m+j-1$. Next, we add an outgoing edge from every node $Q_i$ to every node $x_j\in Q_i$, and an outgoing edge from every node $Q'_i$ to the corresponding node $Q_i$. 

Let $B$ be a best response for agent $u$. First, observe that every node $Q'_i$ does not have any incoming edge, therefore agent $u$ has to buy an edge to each $Q'_i$ in order to have a greedy routing path to it. Since $ d(Q'_i, Q_i) = m  <   2m+i-1= d(u, Q_i)$, the path $(u, Q'_i, Q_i)$ is a greedy routing path from $u$ to $Q_i$. Therefore, by building the edges to every $Q'_i$, agent~$u$ has a greedy routing path to every $Q'_i$ and $Q_i$. Note, that a path via $Q'_i$ and $Q_i$ cannot be a greedy routing path to any node $x_j$.

Next we show that the best response $B$ can be modified so that it does not contain any edge to a node $x_j$. For this, assume that $B$ contains the edge $(u,x_j)$. Since $\{x_1,\dots,x_n\} \subseteq Q_1 \cup \dots \cup Q_m$, there exists a node~$Q_k$ having a directed edge to $x_j$. Moreover, we have $d(Q_k, x_j) = 2m+j-1 < 4m+j-1 = d(u, x_j)$ and thus agent~$u$ could exchange the edge $(u,x_j)$ with the edge $(u,Q_k)$ to obtain the greedy routing path $(u, Q_k, x_j)$  to node~$x_j$. Since node~$x_j$ does not have any outgoing edges, it does not enable greedy routing for agent~$u$ to any other node. Thus, exchanging edge $(u,x_j)$ with $(u,Q_k)$ yields a strategy that enables greedy routing for agent $u$ and that has the same cost as strategy $B$. 
Thus, from now on we can assume that the best response $B$ has no edges to nodes $x_j$.

Let $C$ be a minimum set cover of $\{x_1,\dots,x_n\}$. We show that the number of edges that agent~$u$ builds in a best response $B$ that contains no edges to nodes $x_j$ in the corresponding metric instance is exactly $|C|+m$.

Let $B' = \{Q \mid (u,Q)\in B\}$, i.e., $B'$ is the set of nodes to which agent~$u$ buys an edge under best response $B$. First, observe that $B'\setminus\{Q'_1,\dots, Q'_m\}$ is a subset of $\{Q_1,\dots,Q_m\}$ that covers the set $\{x_1,\dots,x_n\}$. This is true, since otherwise agent~$u$ would not be greedy connected, in particular, at least one node $x_j$ would not be reachable via a greedy routing path. 

Next we show that $B'\setminus\{Q'_1,\dots, Q'_m\}$ is indeed a minimum set cover of $\{x_1,\dots,x_n\}$. Suppose towards a contradiction that there exists a set cover $C$ of $\{x_1,\dots,x_n\}$ such that $|C|< |B'|- m$. Then, buying edges to all the nodes in $C\cup \{Q'_1,\dots, Q'_m\}$ would be a better response for agent~$u$ than strategy~$B$, since agent~$u$ still has a greedy routing path to every $Q_i$ and $Q'_i$ node via the edges to each $Q'_i$ and, by building the edges to the nodes in $C$, agent~$u$ will have a greedy routing path to every $x_j$ node, a contradiction. 

With the above exchange argument, it follows that any best response of agent~$u$ can be converted in polynomial time into a best response $B$ that does not buy edges to $x_j$ nodes. From this, the corresponding set $B' = \{Q \mid (u,Q)\in B\}$ can be efficiently computed. Finally, $B' \setminus \{Q'_1,\dots,Q'_m\}$ is a minimum set cover of the initial set cover instance. Thus, if we could compute a best response for
agent~$u$ in polynomial time, then we could compute a minimum set cover in polynomial time.
\end{proof}
\paragraph*{\textbf{Euclidean Metrics.}} While computing a best response is hard for Euclidean metrics, we now prove the contrasting result that in Euclidean metrics, computing a NE, a social optimum, and deciding if a given strategy profile is a NE is efficiently computable.
\begin{restatable}{corollary}{betacolfour}
    In Euclidean metrics of dimension $D>0$, computing a social optimum, and thus also a NE, is polynomial time computable.
\end{restatable}
\begin{proof}
    By \Cref{independent}, we have that the strategy of an agent~$u$ as long as agent~$u$ is greedy connected does not influence the strategies of the other agents. Moreover, by \Cref{sixedges},  we have that computing a minimum greedy routing set is polynomial time computable. By definition of a minimum greedy routing set, no agent can be greedy routing connected by building fewer edges. Therefore, computing a minimum greedy routing set for every agent~$u$ and building the corresponding edges yields a NE. Furthermore, by \Cref{socialoptimum}, we have that every NE is a social optimum and vice versa. Thus, both can be computed in polynomial time.
\end{proof}
Since in a NE every agent must buy the edges in its minimum greedy routing set, and because this can be computed efficiently, we get the following. 
\begin{corollary}
In Euclidean metrics of dimension $D>0$, deciding if a strategy profile is a NE is polynomial time computable.
\end{corollary}

\section{Greedy Routing: Undirected Edges}
In this section we explore the much more challenging undirected variant of our model. We start with results on two-dimensional Euclidean metrics. Later, we will generalize some of our findings to higher dimensions and general metric spaces.
\subsection{Price of Anarchy}
\paragraph*{\textbf{2D Euclidean Metric.}}
We start with a technical lemma that will be essential for deriving upper bounds on the price of anarchy. 
\begin{restatable}{lemma}{gamalemasix}
\label{lem_non_NNG_edges}
Let $G(\textbf{s})$ be a NE network over a set of points $\mathcal{P}$ and let  $G_{NNG}(\mathcal{P})$ be the corresponding NNG. Let $C\subseteq\mathcal{P}$such that all nodes in $C$ belong to the same connected component of $G_{NNG}(\mathcal{P})$. Then, the total number of edges  built by the agents in $C$ in $G(\s)$ that do not appear in $G_{NNG}$ is at most $2+3|C|$.
\end{restatable}
\begin{proof}  
\begin{figure}[h]
    \centering
\includegraphics[width=0.7\linewidth]{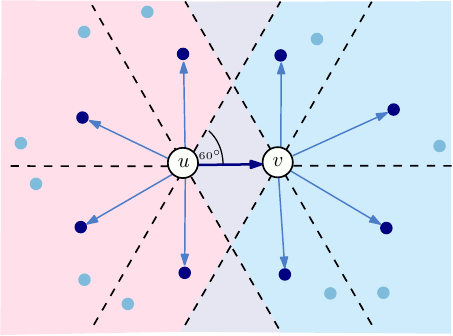}
    \caption{A NNG component of two nodes $u,v$. The maximum number of edges that have to be built in order for agents $u$ and $v$ to be greedy routing connected is 8.}
    \label{fig:neighbor_cone}
\end{figure} 
Let $u,v\in\mathcal{P}$ be two nodes such that $u$ has $v$ as its nearest neighbor in $G_{NNG}(\mathcal{P})$. We partition the space around $u$ into six cones $X_0, X_1,\dots, X_5$ each with an angle of $\frac{\pi}{3}$, arranged counterclockwise starting from the line passing through $u$ and $v$. Since $v$ is the nearest neighbor of $u$, it holds that $d(u, w)\geq d(u,v)$, for all $w\in\mathcal{P}$, with equality to apply only when $\angle{wuv}=\frac{\pi}{3}$. See \Cref{fig:neighbor_cone}.

First, we will consider the case where every node has exactly one nearest neighbor. By applying \Cref{cosine-rule}, we have that for every node~$w\in\mathcal{P}$ lying within cones $X_0, X_5$ there is a greedy routing path from $u$ to $w$ via $v$. Moreover, the same lemma implies that it suffices for $u$ to have at most one edge into each of the remaining cones $X_1, X_2, X_3, X_4$ to ensure that there is a greedy routing path from $u$ to every other node of the network. Thus, besides the edge~$\{u,v\}$  agent~$u$ needs at most four additional edges. 

Next, we partition the space around $v$ into six cones $X_0', X_1',\dots, X_5'$ each with an angle of $\frac{\pi}{3}$, arranged clockwise this time, starting from the line passing through $u$ and $v$ ($X'_0$ overlaps with $X_0$ and $X_1$).

We now show that for every $p\in\mathcal{P}$ lying within the cones $X'_0,X'_5$ there is a greedy routing path from $v$ to $p$ via $u$.
Suppose towards a contradiction that there is a node $p\in\mathcal{P}$ that lies within the corresponding cones $X'_1, X'_5$ such that $p$ is not reachable via greedy routing from $v$ through $u$, i.e., $d(v,p)\leq d(u,p)$. Then, by \Cref{cosine-rule}, it follows that $d(u,p)\leq d(u,v)$, which contradicts the assumption that $u$ has $v$ as its nearest neighbor. Using the same argument as before, we get that $v$ also needs at most four additional edges to ensure that there is a greedy routing path from $v$ to every other node of the network. Therefore, a nearest neighbor component $C$ consisting of two nodes needs at most $8$ edges, i.e., $3|C|+2$ many.

We now consider components with more than two nodes. Suppose $x\in\mathcal{P}$ has node~$v$ as its nearest neighbor. We define new cones $Q_0,\dots, Q_5$, around $v$, and  $Q'_0,\dots, Q'_5$ around $x$, starting from the line passing through $x$ and $v$. By the same argument as before, $x$ needs at most four additional edges, and there is a greedy routing path from $v$ to any node $y$ within $Q_0, Q_1$ via $x$. As shown above, node~$x$ cannot lie within $X'_0, X'_1$. Thus, it must lie in another cone. Consequently, one of the cones $Q_0, Q_5$ overlaps with at least one cone from $X_1,\dots, X_4$. This reduces the number of additional edges $v$ needs to remain greedy connected by at least one.

The case where a node has more than one nearest neighbor is equivalent as long as these neighbors are not also each other's  nearest neighbors. Otherwise, for example if $x$ has also $u$ as a nearest neighbor, then the number of edges that $u$ needs would be also decreased by one. Since we aim to compute an upper bound on the number of additional edges, we  exclude this case from our calculations.  
 
Extending this argument, we get that any node in a connected component that has only one neighbor in $G_{NNG}(\mathcal{P})$ needs at most four additional edges, while a node with at least two such neighbors needs at most three edges. Therefore, any connected component $C$ of size $|C|$  needs at most $3|C|+2$ additional edges. 
\end{proof}

We will now use \Cref{lem_non_NNG_edges} to prove one of our main results: That NE networks are close to optimal in terms of density. In the worst case, selfish behavior only increases the number of edges by 80\%. 
\begin{theorem}
In 2D Euclidean metrics, the PoA is at most $1.8$. \label{poa-2d}
\end{theorem}
\begin{proof}
\begin{figure}[ht]
    \centering
    \includegraphics[width=\linewidth]{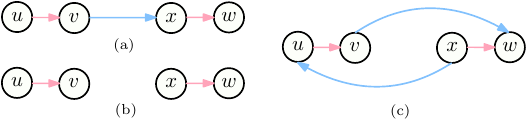}
    \caption{The (a) SO, (b) NNG, and  (c) a NE of a network with four nodes. Directed arcs show the edge-ownership.}
    \label{fig:ne_so_nng}
\end{figure}
 As stated above, any network that supports greedy routing must contain the NNG as a subgraph. Thus, every social optimum and every NE network also contains the NNG as a subgraph. See \Cref{fig:ne_so_nng} for an example.
 
  The NNG consists of connected components, each containing at least two nodes. Consider a NE network $G(\mathbf{s})$ defined over a set of points $\mathcal{P}$, and let $G_{NNG}(\mathcal{P})$ be the corresponding NNG. Partition $\mathcal{P}$ into $k$ disjoint sets $C_1, C_2,\dots, C_k$, 
  where $k$ is the number of connected components of $G_{NNG}(\mathcal{P})$. Each set $C_i$ contains the nodes of the $i$-th connected component, meaning that any two nodes $u,v\in C_i$ are in the same set $C_i$ if and only if they belong to the same connected component of $G_{NNG}(\mathcal{P})$. Then, the social cost is equal to:
   \begin{align*}
   c(\mathbf{s}) &= \sum_{i=1}^{k}\sum_{u\in C_i}|S_u(\s)| = \sum_{i=1}^{k}\sum_{u\in C_i}(|S_u^-(\s)|+|S_u^+(\s)|), \end{align*}
where $S_u^+(\s)$ denotes the set of edges built by agent~$u$ in $G(\mathbf{s})$ that also appear in $G_{NNG}(\mathcal{P})$, and $S_u^-(\s)$ denotes the edges built by agent~$u$ that do not appear in $G_{NNG}(\mathcal{P})$. Equivalently, we define $deg^+_{G(\textbf{s})}(u)$ to be the number of edges incident to $u$ in $G(\textbf{s})$ that also appear in $G_{NNG}(\mathcal{P})$ and $deg^-_{G(\textbf{s})}(u)$ the number of edges incident to $u$ in $G(\textbf{s})$ that do not appear in $G_{NNG}(\mathcal{P})$.

For every agent $u$ there is a set $W_u\subseteq \mathcal{P}$ of agents that agent~$u$ can reach via greedy routing paths starting with an edge from $G_{NNG}(\mathcal{P})$ incdident to agent~$u$. Let $\phi'(u)$ denote the minimum number of edges that $u$ needs to be incident to in order to be able to reach all other agents, i.e., the set $\P \setminus W_u$, via greedy routing paths.

In the worst case, agent~$u$ builds all $\phi'(u)$ many edges by itself. Furthermore, by \Cref{lem:phi_edges} and \Cref{re:consant_phi}, we know that agent~$u$ will not build more than $\phi'(u)$ non-NNG edges. 
Therefore,
\begin{align}
    |S^-_u(\s)| &\leq \phi'(u). \label{phi_1}
\end{align}
     Let $\s^*$ be a strategy profile that achieves the social optimum, and let $G(\mathbf{s^*})$ be the corresponding network. Not counting the edges from $G_{NNG}(\mathcal{P})$, every agent $u$ in $G(\mathbf{s^*})$ must be incident to at least $\phi'(u)$ edges in order to be greedy routing connected. 
     Thus, 
 \begin{align}
     \sum_{u\in \mathcal{P}}{\phi'(u)} &\leq \sum_{u\in \mathcal{P}}{deg^-_{G(\mathbf{s^*})}(u)} = 2\sum_{u\in \mathcal{P}}|S_u^-(\mathbf{s^*})|.\label{phi_2}
 \end{align}

 Combining \eqref{phi_1} and \eqref{phi_2} we get:
 \begin{align}
       \sum_{u\in \mathcal{P}}{\phi'(u)} &\leq 2\sum_{u\in \mathcal{P}}|S_u^-(\mathbf{s^*})|\nonumber\\
     \frac{1}{2}\sum_{u\in \mathcal{P}}{|S^-_u(\mathbf{s})| } &\leq \sum_{u\in \mathcal{P}}|S_u^-(\mathbf{s^*})| \nonumber \\
     \frac{1}{2}\sum_{u\in \mathcal{P}}{|S^-_u(\mathbf{s})| } + \sum_{u\in \mathcal{P}}|S_u^+(\mathbf{s})|&\leq \sum_{u\in \mathcal{P}}|S_u^-(\mathbf{s^*})|\nonumber\\ & \quad\quad+\sum_{u\in \mathcal{P}}|S_u^+(\mathbf{s^*})|,
     \label{phi_3}
 \end{align}
where the last inequality holds since $\sum_{u\in \mathcal{P}}|S_u^+(\mathbf{s^*})| = \sum_{u\in \mathcal{P}}|S_u^+(\mathbf{s})|$, as in both strategy profiles the same number of NNG edges must be bought. Therefore, for the PoA, we have:
\begin{align}
    PoA \leq \frac{c(\mathbf{s})}{c(\mathbf{s^*})} = \frac{\sum_{u\in \mathcal{P}}{|S_u(\mathbf{s})|}}{\sum_{u\in \mathcal{P}}{|S_u(\mathbf{s^*})|}}. \nonumber
\end{align}
 Using \eqref{phi_3} and summing over the connected components we get:
 \begin{align}
    PoA &\leq \frac{ \sum_{u\in \mathcal{P}}{|S^-_u(\mathbf{s})|}+\sum_{u\in \mathcal{P}}{|S^+_u(\mathbf{s})|} }{ \frac{1}{2}\sum_{u\in \mathcal{P}}{|S^-_u(\mathbf{s})|}+\sum_{u\in \mathcal{P}}{|S^+_u(\mathbf{s})|} } \label{eqpoa}\\
    &= \frac{ \sum_{i=1}^{k}\sum_{u\in C_i}{|S^-_u(\mathbf{s})|}+\sum_{i=1}^{k}\sum_{u\in C_i}{|S^+_u(\mathbf{s})|} }{ \frac{1}{2}\sum_{i=1}^{k}\sum_{u\in C_i}{|S^-_u(\mathbf{s})|}+\sum_{i=1}^{k}\sum_{u\in C_i}{|S^+_u(\mathbf{s})|} }. 
\end{align}

The PoA is maximized when $\sum_{i=1}^{k}\sum_{u\in C_i}{|S^-_u(\mathbf{s})|}$ is as large as possible and $\sum_{i=1}^{k}\sum_{u\in C_i}{|S^+_u(\mathbf{s})|}$ is as small as possible. Therefore, using \Cref{lem_non_NNG_edges}, we have $$\sum_{i=1}^{k}\sum_{u\in C_i}{|S^-_u(\mathbf{s})|}\leq \sum_{i=1}^{k}\sum_{u\in C_i}{(3|C_i|+2)},$$ while $$\sum_{i=1}^{k}\sum_{u\in C_i}{|S^+_u(\mathbf{s})|} \geq \sum_{i=1}^{k}\sum_{u\in C_i}{(|C_i|-1)},$$ since $C_i$ is connected. Therefore, we have:
\begin{align}
    PoA &\leq  \frac{ \sum_{i=1}^{k}{(3|C_i|+2)}+\sum_{i=1}^{k}{(|C_i|-1)} }{ \frac{1}{2}\sum_{i=1}^{k}{(3|C_i|+2)}+\sum_{i=1}^{k}{(|C_i|-1)} } \nonumber\\
    &=  \frac{ \sum_{i=1}^{k}{(4|C_i|+1)}}{ \frac{1}{2}\sum_{i=1}^{k}{(3|C_i|+2)}+\frac{2}{2}\sum_{i=1}^{k}{(|C_i|-1)} } \nonumber\\
    &=  \frac{ 2\sum_{i=1}^{k}{(4|C_i|+1)}}{ \sum_{i=1}^{k}{(3|C_i|+2)}+\sum_{i=1}^{k}{(2|C_i|-2)} } \nonumber\\
    &=  \frac{ 8|\mathcal{P}|+2k}{5|\mathcal{P}| }.\label{poa18} 
\end{align}
The maximum number of connected components in the NNG is $\frac{|\mathcal{P}|}{2}$ since every connected component contains at least two nodes.  Therefore,
\begin{align*}
    PoA \leq \frac{(1+8)|\mathcal{P}|}{5|\mathcal{P}|} = 1.8. \quad\quad \quad \qedhere
\end{align*}
\end{proof}
Now we proceed to show that our PoA upper bound is almost tight. For this, we give an almost matching lower bound, showing that selfish behavior can increase the number of built edges by $75\%$.
\begin{lemma}
    In 2D Euclidean metrics the PoA is at least $1.75$.
\end{lemma}
\begin{proof}[Proofsketch]
       \begin{figure}[ht]
        \centering
        \includegraphics[width=0.7\linewidth]{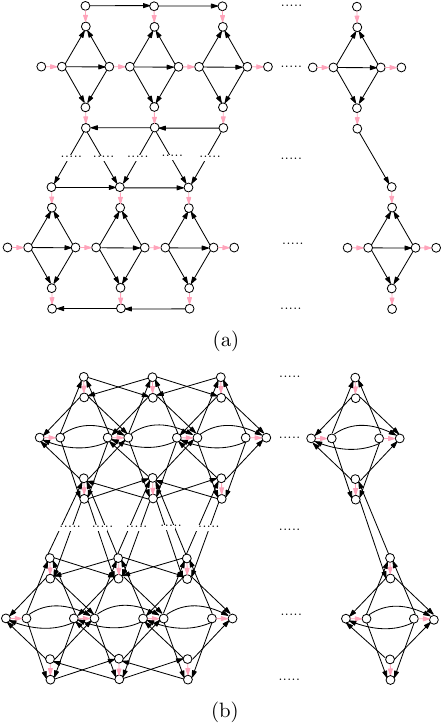}
        \caption{(a) social optimum (SO), (b) a NE with social cost $1.75$ times the cost of SO. Red edges belong to the NNG. All edges are undirected, the arcs depict edge-ownership.}
        \label{fig:lowerbound}
    \end{figure} 
    
Consider the network in \Cref{fig:lowerbound}(a). Let $\mathcal{C}$ denote the set of connected components of the corresponding NNG of the instance. Each NNG component $C \in \mathcal{C}$ consists of two nodes (connected with a red edge) and is connected to nodes outside of~$C$ via 6 edges, the minimum number of required edges for greedy connectivity in this instance. This is a SO, since every non-NNG edge originating from an agent $u$ has equal length and contributes to the greedy routing connectivity of both endpoints. Thus, the total number of edges is $|\mathcal{C}|(\frac{6}{2}+1) = 4|\mathcal{C}|$. In contrast, the network in \Cref{fig:lowerbound}(b), is a NE, where each non-NNG edge contributes to the greedy connectivity of only one endpoint. Hence, the number of the non-NNG edges is doubled compared to the SO, i.e,. it is $|\mathcal{C}|(6+1) = 7|\mathcal{C}|$. Therefore, the social cost ratio is $\frac{7}{4}=1.75$.
\end{proof}

\paragraph*{\textbf{Higher Dimensions  and General Metrics.}} We extend our analysis to obtain upper bounds on the PoA for higher dimensional Euclidean spaces and for general metrics.
\begin{restatable}{corollary}{gammacorfour}
    In Euclidean metrics of dimension $D$, the PoA is at most $2-\frac{1}{K(D)}$, while in general metric spaces it is less than 2.\label{poageneral}
\end{restatable}
\begin{proof}
    We continue our analysis of inequality $\eqref{eqpoa}$ from the proof of~\Cref{poa-2d}:
     \begin{align}
    PoA &\leq \frac{ \sum_{u\in \mathcal{P}}{|S^-_u(\mathbf{s})|}+\sum_{u\in \mathcal{P}}{|S^+_u(\mathbf{s})|} }{ \frac{1}{2}\sum_{u\in \mathcal{P}}{|S^-_u(\mathbf{s})|}+\sum_{u\in \mathcal{P}}{|S^+_u(\mathbf{s})|} } \nonumber\\
    &= \frac{ 2(\sum_{u\in \mathcal{P}}{|S^-_u(\mathbf{s})|}+\sum_{u\in \mathcal{P}}{|S^+_u(\mathbf{s})|}) }{ \sum_{u\in \mathcal{P}}{|S^-_u(\mathbf{s})|}+2\sum_{u\in \mathcal{P}}{|S^+_u(\mathbf{s})|} }\nonumber \\&= 2 - \frac{2\sum_{u\in \mathcal{P}}{|S^+_u(\mathbf{s})|}}{ \sum_{u\in \mathcal{P}}{|S^-_u(\mathbf{s})|}+2\sum_{u\in \mathcal{P}}{|S^+_u(\mathbf{s})|} } \label{poa-eq2}
    \end{align}

    The PoA is maximized when $\sum_{u\in \mathcal{P}}{|S^-_u(\mathbf{s})|}$ is maximized and $\sum_{u\in \mathcal{P}}{|S^+_u(\mathbf{s})|}$ is minimized. In the general metrics case, the maximum number of non-NNG edges is at most:
    $$\sum_{u\in \mathcal{P}}{|S^-_u(\mathbf{s})|}\leq\frac{2(n-1)}{2}-\frac{n}{2}$$
    where the first term is the number of edges of the complete graph and the second term is the minimum number of edges of the NNG. Therefore, the PoA in general metrics is:
    \begin{align}
    PoA &\leq 2 - \frac{2\sum_{u\in \mathcal{P}}{|S^+_u(\mathbf{s})|}}{ \sum_{u\in \mathcal{P}}{|S^-_u(\mathbf{s})|}+2\sum_{u\in \mathcal{P}}{|S^+_u(\mathbf{s})|} } \nonumber\\ &\leq 2-\frac{2\frac{n}{2}}{\frac{n(n-1)}{2}-\frac{n}{2}+2\frac{n}{2}}\nonumber
=  2-\frac{2n}{n(n-1)+n}\\
&\leq 2-\frac{2n}{n^2}=2-\frac{2}{n}<2.
    \end{align}
    For Euclidean metrics of dimension $D>0$, the maximum number of non-NNG edges that an agent builds is at most $K(D)-1$, because every agent has at least one nearest neighbor. Using \eqref{poa-eq2}, we have:
     \begin{align}
     PoA &\leq 2 - \frac{2\sum_{u\in \mathcal{P}}{|S^+_u(\mathbf{s})|}}{ \sum_{u\in \mathcal{P}}{|S^-_u(\mathbf{s})|}+2\sum_{u\in \mathcal{P}}{|S^+_u(\mathbf{s})|} } \nonumber\leq 2-\frac{2\frac{n}{2}}{(K(D)-1)n+2\frac{n}{2}}\nonumber\\
     &\leq 2-\frac{n}{K(D)n} = 2-\frac{1}{K(D)},\nonumber
    \end{align}
    which completes the proof.
\end{proof}

\subsection{Approximate Nash Equilibria}
\paragraph*{\textbf{General Metrics.}} We first study the simplest way for finding a Nash equilibrium: best response dynamics. Unfortunately, these are not guaranteed to converge.
\begin{restatable}{theorem}{gamatheorbest}\label{thm_BRC}
Best response cycles exist. 
\end{restatable}
\begin{proof}
\begin{figure}[h]
    \centering
    \includegraphics[width=0.8\linewidth]{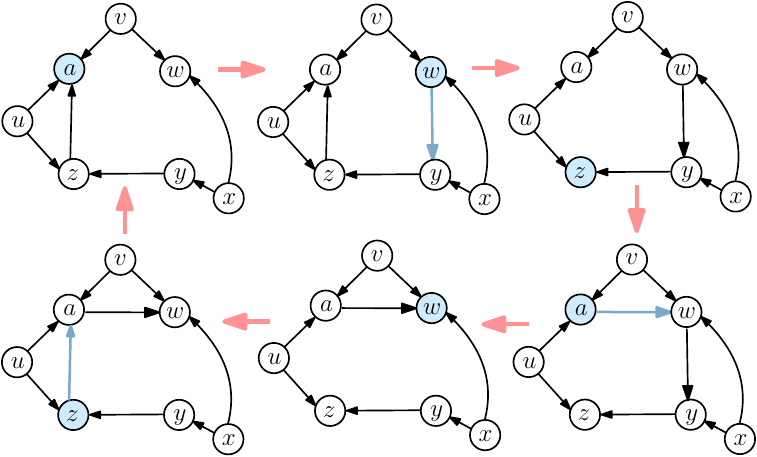}
    \caption{Best Response Cycle.}
    \label{fig:bestresponcecycle1}
\end{figure}
Consider the initial configuration of the network in \Cref{fig:bestresponcecycle1}~(top left). Agent $w$ has a greedy routing path to every agent in the network except agent $z$. Since $y$ is closer to $z$ than $w$ is, agent~$w$ builds an edge to $y$. Next, agent~$z$ is activated, and since it has a greedy routing path to every agent without the edge $\{z,a\}$ it reduces its cost by dropping the edge. The next agent that is activated is agent~$a$, which has a greedy routing path to every agent except $y$. Therefore, agent~$a$ builds an edge to node~$w$, which is closer to $y$ than $a$ is. Then, agent~$w$ is activated and since it has a greedy routing path to every agent without the edge $\{w,y\}$, it reduces its cost by dropping the edge. As a result, agent~$z$ no longer has a greedy routing path to node~$w$. When agent~$z$ is activated next, it builds an edge to node~$a$ that is closer to $w$ than $z$ is. Finally, agent~$a$ is activated and it removes the edge to node~$w$ since it has a greedy routing path to every other agent without this edge, returning the network to the original configuration.
\end{proof}

Next, we introduce some definitions and results that are essential for proving the existence of approximate NE.

\begin{definition}
    Consider graph $G(\P, E)$ that supports greedy routing. For an agent~$u\in \P$ with incident edges $E_u\subseteq E$, the \emph{critical incident set} $H_u\subseteq E_u$ is a maximal set of incident edges such that removing any single edge from $H_u$ breaks agent~$u$'s greedy connectivity.
\end{definition}
\begin{definition}
  Let $G(\P, E)$ be a graph that supports greedy routing and let $H_u$ be a critical incident set of agent~$u$. Consider the modified network $G'=G(\P, E\setminus H_u$), where the edge set $H_u$ is removed. The \emph{critical best response} strategy is defined as agent $u$ 's best response strategy $S_u^{best}$ in $G'$. If multiple best responses exist, we select the one that maximizes the overlap with the critical incident set, i.e., the one for which the quantity~$|S_u^{best}\cap H_u|$ is maximized.
\end{definition}
Next, we define the function $\alpha(u)$, for every agent~$u \in \mathcal{P}$, which reflects how the network's structure influences an agent's strategy.
\begin{definition}
   Let $G(\P, E)$ be a graph that supports greedy routing. Suppose that exactly $k\geq 0$ edges in $S_u^{best}$ are new (not incident to node~$u$ in $G(\P,E)$). Then $\alpha(u) \coloneq k$ is defined as the number of those edges, i.e., $\alpha(u)\coloneq |S_u^{best}\setminus H_u|$. (See \Cref{example_alpha}.) 
\end{definition}
\begin{example}\label{example_alpha}
\begin{figure}[h]
    \centering
\includegraphics[width=\linewidth]{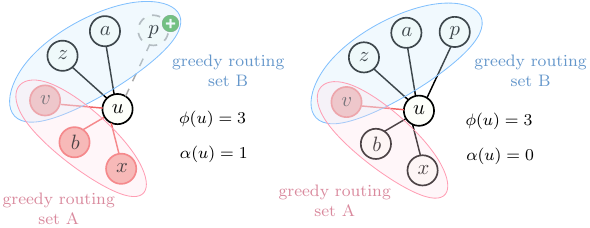}
    \caption{Illustration of the function $\alpha(u)$.}
    \label{fig:alpha_example}
\end{figure}
We want to determine the owner of each edge in \Cref{fig:alpha_example} assuming that the shown network is part of a NE network. 

The sets $H_u, S_u^{best}$ and parameter $\alpha(u)$ are used as tools to achieve this goal. In \Cref{fig:alpha_example}~(left), assume that agent~$u$ is greedy connected under greedy routing set~A, while $H_u=\{\{u, v\}, \{u, b\}, \{u, x\}\}$. Therefore, agent~$u$ has no incentive to buy the edges $\{u,z\},\{u,a\}$, as they do not contribute to its greedy connectivity. Thus, if this was a NE, these edges must therefore be bought by agents~$z$ and $a$, respectively. Moreover, assume that agent~$u$ would be greedy connected under greedy routing set B if there were an edge to node~$p$. Since, the critical best response $S_u^{best}$ is defined as the best response in the network without the $H_u$ edges, and the edges $\{z,u\}, \{a,u\}$ are bought by $z$ and $a$, respectively, we have that $S_u^{best} = \{\{u, v\}, \{u, p\}\}$. As node~$v$ is already incident to $u$, we have $\alpha(u) = 1$.  We can determine which edges agent~$u$ is willing to buy in this instance. First, agent $u$ has an incentive to buy the edge~$\{u,v\}$ because it belongs to both greedy routing sets, i.e., to $S_u^{best}\cap H_u$. Additionally, agent~$u$ can buy at most one edge from $H_u\setminus S_u^{best}$, because $\alpha(u) = 1$.  Otherwise, agent~$u$ could drop the edges $\{u,b\}$ and $\{u,x\}$ and instead buy an edge to node~$p$, reducing its overall cost, since the edges $\{u,z\}$ and $\{u,a\}$ are bought by $z$ and $a$, respectively. Note that while the edges $\{u,a\}$ and $\{u,z\}$ do not contribute to the greedy connectivity of agent~$u$, they still affect the value of $\alpha(u)$.

In \Cref{fig:alpha_example}~(right) the greedy degree of agent~$u$ is $3$ and agent~$u$ is simultaneously greedy connected by two different greedy routing sets, while $H_u = \{\{u, v\}\}$. Thus, $S_u^{best} = \{\{u, v\}\}$, i.e., $\alpha(u) = 0$. Therefore, agent~$u$ has an incentive to buy only the edge $\{u,v\}$. \hfill $\triangleleft$
\end{example}

\begin{lemma}
\label{alphalemma}
Let $G=(\mathcal{P}, E)$ be a graph that supports greedy routing, and let $u \in \mathcal{P}$ be an agent. In any Nash equilibrium of $G$, agent $u$ is  willing to buy only the edges in $S_u^{best}\cap H_u$, along with exactly $\alpha(u)$ many edges from  $H_u\setminus S_u^{best}$. 
\end{lemma}
\begin{proof}
Given a graph $G(\P, E)$ that supports greedy routing, suppose there exists an ownership assignment such that $G$ is in a NE. Our goal is to determine which incident edges an agent is willing to buy in order for $G$ to remain in a NE.

In a NE every agent~$u$ buys a subset $S_u(\textbf{s})\subseteq E_u$ of its incident edges, while the other edges $E_u\setminus S_u$ must be bought by its neighbors. 

First, we  show that in a NE, an agent $u$  does not buy any edges in $E_u\setminus H_u$. Suppose that there is no set $W\subseteq (E_u\setminus H_u)$ such that while the removal of any single edge of them does not affect the greedy connectivity  of $u$, the removal of all the edges in $W$ does. In this case, agent $u$ has no incentive to buy any edge in $E_u\setminus H_u$ since they don't contribute to its greedy connectivity. In the case that such a set $W$ exists, agent $u$ must be greedy routing connected in $G$ by more than one strategy. Otherwise, the removal of any single edge in $W$ would affect the greedy connectivity of agent $u$, which would imply $W_u\subseteq H_u$ (see \Cref{example_alpha}). Suppose that agent~$u$ buys a set of edges in $E_u\setminus H_u$ that belong only to one strategy that enables greedy routing for agent~$u$. Since the rest of the edges are bought by its neighbors, along with them will be the edges that belong to another strategy that enables greedy routing for agent~$u$. Thus, agent~$u$ would be greedy connected by more than one strategy and could reduce its cost by dropping all of its edges. Alternatively, suppose that agent~$u$ buys a set of edges in $E_u\setminus H_u$ that belong to more than one greedy routing strategy. Since all the incident edges are bought by some agent, agent~$u$ could reduce its cost by keeping only the edges of one strategy and dropping the rest. Thus, in every case agent~$u$ will not buy any edge in $E_u\setminus H_u$ when $G$ is in NE.

On the other hand, agent~$u$ would be willing to buy some of the edges in $H_u$ in order to remain greedy routing connected. In any NE, agent~$u$ is willing to buy all the edges from $S_u^{best}\cap H_u$, since these edges are part of the best response. In addition, agent~$u$ is willing to buy up to $\alpha(u)$ many edges from the remaining edges in $H_u$, i.e., from $H_u\setminus S_u^{best}$. Otherwise, agent~$u$ could reduce its cost by removing the edges in $H_u$ that do not appear in $S_u^{best}\cap H_u$ and instead buying these $k$ edges. This completes the proof.
\end{proof}

In a NE graph $G(\s)$ every edge contributes to the greedy connectivity of at least of one of its endpoints. We distinguish between two edge types, \emph{single-edges} and \emph{double-edges} defined as follows:
\begin{definition}
    In a NE graph $G(\s)$ the strategy of an agent $u\in \P$ can be partitioned into two sets: the set of \emph{single-edges} $S_u^s$ and the set of \emph{double-edges} $S_u^d$. An edge $\{u,v\}\in S_u^s$ is called \emph{single-edge} if it contributes to the greedy connectivity of only agent~$u$, while an edge $\{u,v\}\in S_u^d$ is called \emph{double} if it contributes to the greedy connectivity of both endpoints $u$ and $v$. 
\end{definition}
Note that although a single-edge $\{u,v\}$ contributes to the greedy connectivity of only $u$ it may affect the value of $\alpha(v)$ of agent $v$. This can happen when a single-edge $\{u,v\}$ bought by agent~$u$ is required to be present in a best response of $v$. (See \Cref{example_alpha}.)

Another tool that will be useful is Hakimi's theorem \cite{hakimi1965degrees}:

\begin{theorem}
    (Hakimi's Theorem) Let $G(V, E)$ be an undirected graph and let $\delta$ be an integer-valued function on $V$. Then $G$ admits an orientation such that the in-degree of each vertex $u\in V$ is exactly $\delta(u)$, if and only if for every subset $U\subseteq V$ holds $\sum_{u\in U}{\delta(u)}\geq |E(U)|$, where $E(U)$ denotes the set of edges in $G$ with both endpoints in $U$.\label{Hakimi}
\end{theorem}
 We proceed by proving that every social optimum is a 2-NE. 

\begin{theorem}
 In general metric spaces, every SO is a $2$-NE.\label{2ne}
\end{theorem}
\begin{proof}
 The proof is structured as follows. We begin with a SO graph $G = (\mathcal{P},E)$, where no ownership is assigned to its edges. Our goal is to prove that there is a valid ownership assignment such that $G$ in an approximate NE.

To achieve this, we first construct the auxiliary graph $\Tilde{G}(\P, \Tilde{E})$. For each agent~$u\in\P$, the graph $\Tilde{G}(\P, \Tilde{E})$ includes all the single-edges of agent~$u$ as well as the edges that belong to its critical best response $S_u^{best}$, i.e $\Tilde{E}= \bigcup_{u\in\P}{(S_u^s\cup S_u^{best})}$. By the definition of a critical best response, and since all the single-edges are present, this graph supports greedy routing.

Since $G$ is a SO, $\Tilde{G}$ has to have at least as many edges as $G$. Therefore we have:
\begin{align}
    |\Tilde{E}| &\geq |E| \nonumber\\
    \sum_{u\in\P}{(|S_u^{best}\cup S_u^s|)} &\geq |E| \nonumber\\
    \sum_{u\in\P}{(|S_u^{best}\setminus H_u|+|S_u^{best}\cap H_u|+|S_u^{s-}|)} &\geq |E|, \nonumber 
\end{align}
where $S_u^{s-}$ denotes the set of single-edges associated with agent $u$ that are not in $(S_u^{best}\cap H_u)$. Since  $\alpha(u) = |S_u^{best}\setminus H_u|$, we have:

\begin{align}
   \sum_{u\in \P}{(\alpha(u)+|S_u^{best}\cap H_u|+|S_u^{s-}|)}&\geq |E|. \label{hakimieq}
\end{align}
Now that we obtained this inequality, we go back to the SO graph~$G$. By \Cref{alphalemma}, in any NE of $\P$ agent~$u$ is willing to buy all the edges in $S_u^{best}\cap H_u$ and exactly $\alpha(u)$ many edges from $H_u\setminus S_u^{best}$. 

First, we assign all the edges in $S_u^{best}\cap H_u$ to agent~$u$. If an edge belongs to the critical best response of both endpoints we choose one arbitrarily. 
Then,  observe that the number of the remaining single-edges contributing to the greedy connectivity of an agent $u\in\P$ cannot be more than $\alpha(u)$, because otherwise we could replace the edges in $S_u^{s-}$ in $G(\P, E)$ with $\alpha(u)$ many edges from $S_u^{best}\setminus H_u$ without affecting the greedy connectivity of the other agents, creating a graph with less edges than the social optimum, a contradiction. Next, we assign the $S^{s-}_u$ edges (which are at most $\alpha(u)$ many) to agent~$u$, for every $u\in\P$. 

Let $G'(\P, E')$ be the subgraph of 
$G$ containing all remaining unassigned edges  $E' = E\setminus((S_u^{best}\cap H_u)\cup S^{s-}_u)$. Therefore we have:
\begin{align}
\sum_{u\in \P}{(\alpha(u)+|S_u^{best}\cap H_u|+|S_u^{s-}|)}&\geq |E|\nonumber\\
\sum_{u\in \P}{\alpha(u)}&\geq |E|-\sum_{u\in \P}{(|S_u^{best}\cap H_u|+|S^{s-}_u|)}\nonumber\\
   \sum_{u\in \P}{\alpha(u)} &\geq |E'|.\nonumber
\end{align}
 Every edge in $(u,v)\in E'$ contributes to the greedy connectivity of both agents $u$ and $v$, so both of them have an incentive to buy it as soon as their remaining $\alpha$-value permits. 

 Suppose towards a contradiction, that there exists a subset $U\subseteq\P$ such that $\sum_{u\in U}{\alpha(u)} < |E'(U)|$, where $E'(U)$ denotes the set of edges in $G'$ with both endpoints in $U$. In this case, we could replace the edges in $E'(U)$ with the corresponding $\alpha(u)$ many edges in $S_u^{best}\setminus H_u$ for each agent $u$, and then add the assigned edges in $E\setminus E'$ having as a result a graph that supports greedy routing and has fewer edges than the social optimum, a contradiction.

Hence, Hakimi's theorem (\Cref{Hakimi}) applies to $G'(\P, E')$. Let $indeg_{G'}(u)$ denote the in-degree of a node $u\in \P$ in graph $G'$. By Hakimi's theorem, the graph $G'$ can be oriented such that for each $u\in\mathcal{P}$, we have $indeg_{G'}(u)\leq \alpha(u)$. Interpreting the orientation as edge ownership, there exists an ownership assignment in $G'$ such that  agent $u$ buys at most $\alpha(u)$ many edges. Using this assignment in the edges $E\cap E'$ in $G$, we end up with an ownership assignment where every agent $u$ builds $|S_u^{best}\cap H_u|+2\alpha(u)$ many edges, while in a NE it would have built $|S_u^{best}\cap H_u|+\alpha(u)$ many. Therefore, the social optimum network $G$ is in 2-NE.    
\end{proof}

\paragraph*{\textbf{2D Euclidean Metrics.}}
We now focus on 2D Euclidean metrics and compared to \Cref{2ne}, we prove the stronger result that every social optimum is a +2-NE. 
\begin{restatable}{theorem}{gammatheorseven}
In 2D Euclidean metrics, every SO is a +2-NE.\label{+2ne} 
\end{restatable}
\begin{proof}
In order to prove the additive approximation bound we extent the proof of \Cref{2ne}. In particular, let $G(\P, E)$ be a SO, and for each agent $u\in\P$ let $H_u$ denote the critical incident set, let $S_u^{best}$ be the critical best response, and let $S_u^{s-}$ be 
the set of single-edges associated with agent~$u$ that are not in $S_u^{best}\cap H_u$. Then, according to inequality~\eqref{hakimieq} in the proof of \Cref{2ne} we have:
\begin{align}
   \sum_{u\in \P}{(\alpha(u)+|S_u^{best}\cap H_u|+|S^{s-}_u| )}&\geq |E|. \nonumber
\end{align}
By \Cref{alphalemma}, in any NE of $\P$ agent $u$ is willing to buy all the edges in $S_u^{best}\cap H_u$ and exactly $\alpha(u)$ many edges from $S_u^{best}\setminus H_u$. Therefore, we assign to every agent~$u$ all the edges in $S_u^{best}\cap H_u$. 

Next, we will show that for every agent $u\in\P$ it holds
\begin{align}
    \alpha(u)&\geq |S_u^{s-}|, \nonumber 
\end{align}
because otherwise, since the single-edges affect only the greedy connectivity of the corresponding agent~$u$, we could replace the edges in $S_u^{s-}$ with $\alpha(u)$ many edges having as a result a graph that has less edges than the SO. Therefore, this implies:
\begin{align}
     \sum_{u\in \mathcal{P}}{\alpha(u)}&\geq \sum_{u\in \mathcal{P}}{|S_u^{s-}|}. \label{alphaineq}
\end{align}

Since $\alpha(u)\geq |S_u^{s-}|$, we assign to each agent~$u$ all the edges in $S_u^{s-}$. We then construct a graph $G'(\P, E')$, consisting of all the unassigned edges in $E$, i.e., the set $E' = E\setminus\bigcup_{u\in \P}{((S^{best}_u\cap H_u)\cup S_u^{s-})}$. For each agent $u\in\P$ we define the function $\alpha'(u) = \alpha(u)-|S_u^{s-}|$. Hence we have:
\begin{align}
   \sum_{u\in \P}{(\alpha(u)+|S_u^{best}\cap H_u|+|S^{s-}_u| )}&\geq |E|. \nonumber\\
    \sum_{u\in \P}{\alpha(u) }&\geq |E|- \sum_{u\in \P}{(|S_u^{best}\cap H_u|+|S^{s-}_u|)} \nonumber\\
 \sum_{u\in \P}{\alpha(u)}&\geq |E'|. \nonumber
\end{align}

Observe that this inequality holds for every induced subgraph of~$G'$. Otherwise, using the same argument as before, we could replace all the edges of agent~$u$ with the $\alpha(u)$ many edges in $S_u^{best} \setminus H_u$, resulting in a graph with fewer edges than $E'$. Adding the remaining edges in $E\setminus E'$ would then yield a graph supporting greedy routing with fewer edges than the social optimum, a contradiction.

Replacing $\alpha(u)$ with $\alpha'(u)$ we get:
\begin{align}
      \sum_{u\in \P}{(\alpha'(u)+|S_u^{s-}|)}&\geq |E'|. \nonumber
\end{align}

Consider a subset of agents $U\subseteq \P$, where $E'(U)$ denotes the set of edges in $G'$ with both endpoints in $U$. Thus, we have:
\begin{align}
         \sum_{u\in U}{(\alpha'(u)+|S_u^{s-}|)}&\geq |E'(U)|\nonumber\\
         \sum_{u\in U}{\alpha'(u)}+\sum_{u\in U}{|S_u^{s-}|}&\geq |E'(U)|. \nonumber
\end{align}

In order to prove the desired approximation bound, we first  explore the case where  $\sum_{u\in U}{|S^{s-}_u|}\leq 2|U|$. Replacing $\sum_{u\in U}{|S_u^{s-}|}$ with $2|U|$ in the inequality above we get:
\begin{align}
\sum_{u\in U}{\alpha'(u) }+2|U| &\geq |E'(U)|\nonumber\\
        \sum_{u\in U}{(\alpha'(u) +2)} &\geq |E'(U)|.\nonumber
\end{align}

    Next, we prove that the same inequality holds in the case where $\sum_{u\in U}{|S^{s-}_u|} >2|U|$. We start by modifying  inequality~\eqref{alphaineq}:
\begin{align}
    \sum_{u\in U}{\alpha(u)} \geq  \sum_{u\in U}{|S_u^{s-}|}. \nonumber
\end{align}
Replacing $\sum_{u\in U}{|S_u^{s-}|}$ with $2|U|$ we get:
\begin{align}
      \sum_{u\in U}{\alpha(u)} &> 2|U|.\nonumber
\end{align}
Adding $2|U|$ to both sides then yields
\begin{align}
    \sum_{u\in U}{\alpha(u)}+ 2|U| &> 4|U|. \nonumber
\end{align}

Since the Delaunay triangulation supports greedy routing and since it has at most $3|U|-6$ many edges, it follows that $|E(U)|\leq 3|U|-6 \leq 4|U|$. Otherwise, by \Cref{independent}, we could replace the edges in $E(U)$ with those of the Delaunay triangulation, resulting again in a graph that supports greedy routing and has fewer edges than the SO. Therefore:
\begin{align}
\sum_{u\in U}{\alpha(u)}+ 2|U| &> 4|U|>|E(U)|.\nonumber
\end{align}
Replacing $\alpha(u)= \alpha'(u)+ |S_u^{s-}|$, we get:
\begin{align}
\sum_{u\in U}{(\alpha'(u)+|S^{s-}_u|)}+ 2|U|  &>|E(U)|\nonumber\\
\sum_{u\in U}{\alpha'(u)}+ 2|U|  &>|E(U)|-\sum_{u\in U}{|S^{s-}_u|}\nonumber\\
\sum_{u\in U}{(\alpha'(u)+2)} &>|E'(U)|. \nonumber
\end{align}

Combining these two cases yields:
\begin{align}
   \sum_{u\in U}{(\alpha'(u)+2)}&\ge |E'(U)|.
\end{align}
Applying Hakimi's theorem (\Cref{Hakimi}) as before, we can compute an ownership assignment for the unassigned edges in $G'$. Therefore, in NE every agent $u$ buys $|H_u\cap S_u^{best}|+\alpha(u)$ many edges while in this construction an agent~$u$ buys $|H_u\cap S_u^{best}|+\alpha'(u)+|S_u^{s-}|+2 = |H_u\cap S_u^{best}|+\alpha(u)+2$ many. Hence, $G$ is in +2-NE.
\end{proof}
\subsection{Computational Complexity}
\paragraph*{\textbf{General Metrics}}
We now show that computing a best response is NP-hard, while an approximate NE can be computed in polynomial time in Euclidean metrics.
\begin{restatable}{theorem}{gammattheoreight}
Computing a best response is NP-hard.
\end{restatable}
\begin{proof}
    We modify the proof of \Cref{bestresponsehardness} as follows: First we replace every directed edge with an undirected one and we assign the ownership to the tail node of the original edge. Next, we add an edge from each node $Q'_i$ to $u$, assigning the ownership of the edge to $Q'_i$. In this way agent~$u$ has a greedy routing path to every node~$Q'_i$ and every node~$Q_i$. The rest of the proof remains the same.
\end{proof}
\begin{algorithm}
\caption{Compute\_Approximate\_NE }\label{algo}
\begin{algorithmic}[1]  % [1] means line numbers
\State $E \gets \emptyset$
\ForAll{$u\in\P$}
    \State $\Phi(u)\gets$ minimum greedy routing set of agent~$u$
    \State $E\gets E\cup \Phi(u)$ 
\EndFor
\State $X \gets$ edges in $E$ whose removal do not disable greedy routing for any agent
\State $E\gets E\setminus X$
\State  $G \gets G(\P, E)$
\State $flag\gets true$
\While{$flag$}
    \State $A\gets \emptyset, \ A_u\gets \emptyset, \ S \gets \emptyset, \ S_u^{s-} \gets \emptyset$
    \State $H_u\gets$ critical set of agent~$u$ in $G$
    \State $S_u^{best}\gets$ critical best response of agent~$u$ in $G$   
    \ForAll{$u\in\P$}
        \State $A_u\gets S_u^{best}\setminus H_u$
        \State $S_u \gets$ single-edges of agent~$u$ that are not in $S_u^{best}\cap H_u$ 
        \State $A\gets A\cup A_u$
        \State $S\gets S\cup S_u^{s-}$
    \EndFor
    \State $a_{change} \gets false$
    \ForAll{$u\in\P$}
        \If{$|A_u| < |S_u^{s-}|$}
            \State $E \gets (E\setminus S^{s-}_u)\cup A_u$
            \State $a_{change} \gets true$
            \State \textbf{break}
        \EndIf
    \EndFor
    \If{$a_{change}$} 
        \State \textbf{continue}
    \EndIf
   \If{HakimiMaxFlow(A, S, E)}
       \State $w\gets$ ownership assignment
       \State $flag\gets false$
    \Else
        \State $E \gets A\cup S$
    \EndIf
    \State $G\gets G(\P, E)$
\EndWhile 
\State \Return $G, w$
\end{algorithmic}

\end{algorithm}
\paragraph*{\textbf{Euclidean Metrics}}
Here, we prove that an approximate NE with at most 1.8 times the optimal number of edges in 2D Euclidean metrics, and $2-1/K(D)$ times the optimal number of edges in $D$-dimensional Euclidean metrics can be computed in polynomial time. Note that this bound improves on the approximation ratio of 3 achieved by the Delaunay triangulation in 2D Euclidean metrics.

\begin{restatable}{theorem}{gammatheornine}
    A $2$-NE in Euclidean metrics of dimension $D>0$, and a $+2$-NE in 2D Euclidean metric that has at most 1.8 times the optimal number of edges in 2D Euclidean metrics, and $2-1/K(D)$ times the optimal number of edges in $D$-dimensional Euclidean metrics, can be computed in polynomial time.
\end{restatable}
\begin{proof}
   Consider \Cref{algo}. First we compute the minimum greedy routing set $\Phi(u)$ for every agent $u\in\P$. Then, we filter the resulting edges by excluding the ones whose removal does not disable greedy routing for any agent. Next, we construct a graph $G(\P, E)$ with the remaining edges. Since the minimum greedy routing set of each agent has constant size and the filtering step involves checking the incident edges per agent, this procedure runs in polynomial time.

   Next, for a number of iterations which we will show later that is less than $(K(D)-1)|\P|-1$, we compute for every agent $u$ the critical incident set $H_u$, the critical best response $S_u^{best}$, the edge set $A_u=S_u^{best}\setminus H_u$ and the set of single edges that are not in $S_u^{best}\cap H_u$ or $S_u^{s-}$. Each of these computations can be performed in polynomial time. Note that $|A_u| = \alpha(u)$. 
   
   For each agent $u\in \P$ we check whether $\alpha(u)< |S_u^{s-}|$. If this is the case, the edges in $S_u^{s-}$ are replaced with the edges in $A_u$ within $E$. The resulting graph contains fewer edges
   and still supports greedy routing, since the single-edges associated with an agent $u$ affect only its own greedy connectivity. The algorithm then recomputes  $H_u, A_u, A$, and $S_u^{best}$ for each agent $u\in\P$. Because every graph  supporting greedy routing must have at least $|\P|-1$ edges, and each iteration of the loop strictly decreases the number of the edges, this procedure can repeat at most $(K(D)-1)|\P|-1$ times. 
   
   If on the other hand, $\alpha(u)\geq |S_u^{s-}|$ for every $u\in\P$, then we check if Hakimi's theorem (\Cref{Hakimi}) applies, i.e., if $\sum_{u\in\P}{(\alpha(u)+S_u^{s-})}\geq |E'|$, where $E' =  \bigcup_{u\in\P}{(H_u\setminus (S_u^{best}\cup S_u^{s-}))}$. This condition can be checked in polynomial time using a max-flow algorithm \cite{hakimi1965degrees, gale1957theorem}. If the maximum flow is smaller than $|E'|$, the conditions for Hakimi's theorem are not satisfied. In that case, we replace the edges in $E'$ with the edges in $A\cup S$, decreasing the number of the edges and restart the loop. Again, this can happen at most $(K(D)-1)|\P|-1$ times. 
   
   If the conditions for Hakimi's theorem hold, we obtain an ownership assignment for the edges $E'$ using the max-flow solution. Then, we assign the remaining edges in $H_u, S_u$ to each agent $u$. Since, $\alpha(u)\geq S_u^{s-}$ for every agent $u$, and since Hakimi's theorem applies, it follows (as shown in \Cref{2ne}) that the resulting graph is a 2-NE for Euclidean metrics. Furthermore, the total number of edges is at most $2-1/K(D)$ times the optimal. This holds because in the initial construction of $G(\P, E)$ we followed the same procedure used for computing the social cost in \Cref{poa-2d} and \Cref{poageneral}.

   For 2D Euclidean metric, in line 8 we compute the Delaunay triangulation of the same point set. If the Delaunay triangulation has fewer edges than $E$, we replace $E$ with the edges of the Delaunay triangulation. Therefore, the resulting graph has at most  $3|\P|-6$ edges. Using,  \Cref{+2ne} and \Cref{poa-2d} we get that this construction is a +2-NE and the number of edges is at most 1.8 times the optimal. 
\end{proof}
\begin{corollary}
In Euclidean metric of dimension $D>0$, deciding if a strategy profile is a NE is polynomial time computable.
\end{corollary}
\begin{proof}
Since, as shown above, the set of the edges that are bought by each agent $u$ in any NE can be computed in polynomial time, we check that the corresponding edges are assigned to each agent $u$. For the other edges, we check if the remaining $\alpha$-value is at most the number of edges assigned to each agent~$u$.
\end{proof}

\section{Conclusion}
We study decentralized network creation to enable greedy routing. We uncover that for the case with directed edges equilibria exist and are optimal, while the setting with undirected edges is much more challenging. For the latter, we give an algorithm that constructs approximate equilibria with only slightly more edges compared to the social optimum. Moreover, we give an almost tight bound on the price of anarchy, which shows that this increase in the number of edges due to the selfishness of the agents cannot be avoided. However, we emphasize that our constructed networks still significantly outperform the well-known Delaunay triangulation.  

Future work can build on our basis that combines network creation games with computational geometry. In particular, enhancing the construction so that additional features like stretch or robustness guarantees can be achieved. Moreover, proving that exact equilibria exist is an exciting open problem. We believe that our construction might already yield Nash equilibria, but proving this is challenging and requires additional structural insights.  

\bibliographystyle{ACM-Reference-Format} 
\bibliography{references}

\clearpage

\end{document}